\documentclass[a4rpaper,twoside,leqno,twocolumn]{article}

\usepackage{ltexpprt}
\usepackage{hyperref}

\usepackage{graphicx}

\usepackage{amsmath}
\usepackage{amssymb}
\usepackage{amsfonts}
\usepackage{times}
\usepackage{mathptmx} 
\usepackage{datetime}

\newcommand{\sinhc}{\mathrm{sinhc}}
\newcommand{\sinc}{\mathrm{sinc}}
\DeclareMathAlphabet{\bit}{OML}{cmm}{b}{it}
\newcommand{\ad}{\mathrm{ad}}           
\def\<{\leqslant}           
\def\>{\geqslant}           

\def\d{\partial}
\def\wh{\widehat}
\def\wt{\widetilde}

\def\Re{\mathrm{Re}}   
\def\Im{\mathrm{Im}}   

\def\cH{\mathcal{H}}   
\def\mR{{\mathbb R}}    
\def\mC{\mathbb{C}}    

\def\Tr{\mathrm{Tr}}       
\def\rT{{\rm T}}        
\def\supp{\mathrm{supp}}    

\def\lexp{\mathop{\overleftarrow{\exp}}}
\def\rexp{\mathop{\overrightarrow{\exp}}}

\def\bE{\mathbf{E}}    

\def\bra{{\langle}}
\def\ket{{\rangle}}

\def\re{{\rm e}}        
\def\rd{{\rm d}}        

\def\br{\mathbf{r}}
\def\x{\times}
\def\ox{\otimes}

\def\fB{\mathfrak{B}}

\def\fF{\mathfrak{F}}

\def\fH{\mathfrak{H}}

\def\cW{\mathcal{W}}

\def\cC{\mathcal{ C}}
\def\cR{\mathcal{ R}}

\def\sA{\mathsf{A}}
\def\sB{\mathsf{B}}

\def\cI{\mathcal{I}}

\def\ups{\upsilon}
\def\Ups{\Upsilon}

\begin{document}

\title{\Large Lie-algebraic Connections Between Two Classes of Risk-sensitive Performance
Criteria for Linear Quantum Stochastic Systems\thanks{This work is supported by the Air Force Office of Scientific Research (AFOSR) under agreement number FA2386-16-1-4065 and the Australian Research Council under grant DP180101805.}}
\author{
Igor G. Vladimirov$^\dagger$,
\and
Ian R. Petersen$^\dagger$,
\and
Matthew R. James%
\thanks{Research School of Engineering, College of Engineering and Computer Science,
Australian National University, Canberra, 
Australia, {\tt
igor.g.vladimirov@gmail.com, i.r.petersen@gmail.com, matthew.james@anu.edu.au}
}
}
\date{}

\maketitle







\begin{abstract}
\small\baselineskip=9pt
This paper is concerned with the original risk-sensitive performance criterion for quantum stochastic systems and its recent quadratic-exponential  counterpart. These functionals are of different structure because of the noncommutativity of quantum variables and have their own useful features  such as tractability of evolution equations and robustness properties.
We discuss a Lie algebraic connection between these two classes of cost functionals for open quantum harmonic oscillators using an apparatus of complex Hamiltonian kernels and symplectic factorizations.   These results are aimed to extend useful properties from one of the classes of risk-sensitive costs to the other and develop state-space equations for computation and optimization of these criteria in quantum robust control and filtering problems.
\end{abstract}

\section{Introduction}
\setcounter{equation}{0}

Open quantum harmonic oscillators (OQHOs) \cite{GZ_2004}, governed by linear quantum stochastic differential equations (QSDEs),  constitute an important application of the Hudson-Parthasarathy  calculus \cite{HP_1984,P_1992} to the modelling of quantum systems interacting with external bosonic fields. The class of OQHOs is closed under concatenation, and their interconnection into a quantum feedback network  \cite{GJ_2009,JG_2010} (for example, a closed-loop system formed from a plant and controller, both modelled as OQHOs) is also an OQHO whose parameters are expressed in terms of the subsystems.

Quantum control and filtering problems for such systems 
\cite{B_1983,B_2010,BH_2006,BVJ_2007,EB_2005,J_2004,J_2005,JNP_2008,NJP_2009,VP_2013a,VP_2013b,WM_2010} aim to achieve certain dynamic properties for quantum plants by using measurement-based feedback with classical controllers and filters or coherent (measurement-free) feedback   involving direct or field-mediated connection \cite{ZJ_2012} with other quantum systems. 
The performance criteria  combine qualitative requirements (such as stability) with optimality principles in the form of the minimization of cost functionals. In particular, quantum linear quadratic Gaussian (LQG) control and filtering \cite{EB_2005,MP_2009,NJP_2009} are concerned with  minimising the mean square values of the closed-loop system variables, similarly to the  classical LQG control and filtering problems \cite{AM_1989,KS_1972}. 

The quantum risk-sensitive performance criterion, originated in \cite{J_2004,J_2005} for measurement-based quantum control and filtering problems (see also \cite{DDJW_2006,YB_2009}), employs the mean square value of a time-ordered exponential (TOE) driven by a function of the system variables.  This cost functional imposes an exponential penalty on the system variables and involves their multi-point quantum states at different moments of time. Since, even in the Gaussian case \cite{CH_1971,KRP_2010},  such states do not reduce to classical joint probability distributions because of the noncommutativity of quantum variables,  the quantum risk-sensitive cost differs from its classical predecessors \cite{BV_1985,J_1973,W_1981}. Nevertheless, this cost functional allows for tractable evolution equations and an appropriate modification of the information state techniques  in application to the measurement-based quantum control settings.

The structure of the classical risk-sensitive performance criteria (as the exponential moment of a quadratic function of the system variables over a time interval) has recently been adopted in a quadratic-exponential functional (QEF) \cite{VPJ_2018a}. Despite a more complicated evolution (compared to the original quantum risk-sensitive cost), the QEF leads to upper bounds \cite{VPJ_2018a}  for the tail distribution of the corresponding quadratic function of the quantum system variables in the spirit of the large deviations theory \cite{DE_1997,S_1996}. Moreover, the QEF gives rise to guaranteed upper bounds \cite{VPJ_2018b} for the worst-case value of the quadratic cost when the actual quantum state may depart from its nominal model, with the departure being described in terms of the quantum relative entropy \cite{J_2004,OW_2010,YB_2009}. The role of the QEF  in the quantum robust performance estimates is similar to the connections between risk-sensitive control and minimax LQG control for classical stochastic systems with a relative entropy description of statistical uncertainty in the driving noise \cite{DJP_2000,P_2006,PJD_2000,PUS_2000}.

The useful properties can be extended from one of the risk-sensitive costs to the other through bilateral links  between these two classes of quantum performance criteria, 
which is the main theme of the present paper.  
To this end, we develop a continuous-time analogue of the results of \cite{VPJ_2018c}, which leads to a Lie-algebraic correspondence between the QEF and the original TOE-based quantum risk-sensitive cost driven by a quadratic function of the system variables. An important ingredient of this connection is an isomorphism between the Lie algebra of quadratic functions of the system variables of the OQHO with complex symmetric kernels and the Lie algebra of complex Hamiltonian kernels, 
which are infinitesimal generators of complex symplectic kernels 
(all these kernels are matrix-valued).

The paper is organised as follows.
Section~\ref{sec:OQHO} specifies the class of linear quantum stochastic systems under consideration.
Section~\ref{sec:funs} describes the original quantum risk-sensitive cost and its quadratic-exponential counterpart.
Section~\ref{sec:quad} represents a class of quadratic functions of system variables using complex symmetric matrix-valued  measures.
Section~\ref{sec:iso} describes an isomorphism of this class to a Lie algebra of complex Hamiltonian kernels. 
Section~\ref{sec:cont} establishes a Lie-algebraic correspondence between two classes of TOE-based and quadratic-exponential functions of system variables.
Section~\ref{sec:move} represents this correspondence in integro-differential form. 
Section~\ref{sec:spec} discusses the Lie-algebraic correspondence and specific nonanticipative measures for the QEF and TOE-based criteria driven by quadratic functions of the current system variables.
Section~\ref{sec:conc} makes concluding remarks.

\section{Open quantum harmonic oscillators}
\label{sec:OQHO}
\setcounter{equation}{0}

We consider an OQHO with (an even number of) dynamic variables $X_1, \ldots, X_n$ (for example, pairs of conjugate  quantum mechanical positions and momenta \cite{S_1994}). These system variables
are time-varying self-adjoint operators on (a dense domain of) a complex separable Hilbert space $\fH$ and are assembled into a vector $X:= (X_k)_{1\< k\< n}$ 
(vectors are organised as columns,  
and the time arguments are often omitted for brevity). They satisfy the canonical commutation relations (CCRs)
$    \cW_{u+v} = \re^{i u^{\rT}\Theta v} \cW_u \cW_v
$ for all $u,v\in \mR^n$,
where $i:= \sqrt{-1}$ is the imaginary unit, and  $\cW_u:= \re^{iu^{\rT} X}
$ is the unitary Weyl operator \cite{F_1989}. 
Here,
$\Theta$ is a nonsingular real antisymmetric matrix specifying the matrix
\begin{equation}
\label{Theta}
    [X, X^{\rT}]
    :=    ([X_j,X_k])_{1\< j,k\< n}
     =
     2i \Theta \ox \cI_{\fH}
\end{equation}
of commutators $[X_j, X_k]
    :=
    X_jX_k - X_kX_j$
as the Heisenberg infinitesimal form
(on a dense domain in $\fH$) for the Weyl CCRs, with $\ox$ the tensor product, and $\cI_{\fH}$ the identity operator on $\fH$.  The matrix $\Theta \ox \cI_{\fH}$ will be identified with $\Theta$. The system variables of the OQHO evolve in time according to a linear QSDE
\begin{equation}
\label{dX}
    \rd X
    =
    \sA X \rd t+ \sB  \rd W,
\end{equation}
where $\sA\in \mR^{n\x n}$, $\sB  \in \mR^{n\x m}$ are constant matrices whose structure is clarified below.
This QSDE is driven by the vector $W:=(W_k)_{1\< k \< m}$ 
of an even number $m$ of quantum Wiener processes $W_1, \ldots, W_m$ which are time-varying self-adjoint operators on a symmetric Fock space $\fF$ \cite{P_1992,PS_1972}. These operators  represent the external bosonic fields and 
have a complex positive semi-definite Hermitian Ito matrix $\Omega := I_m + iJ\in \mC^{m\x m}$, so that $\rd W \rd W^{\rT}
    =
    \Omega \rd t
$, with $I_m$ the identity matrix of order $m$. Its imaginary part
$        J
        :=
       {\scriptsize\begin{bmatrix}
           0 & I_{m/2}\\
           -I_{m/2} & 0
       \end{bmatrix}}
$
is an orthogonal real antisymmetric matrix of order $m$ (so that $J^2=-I_m$), which specifies CCRs for the quantum Wiener processes as $[W(s), W(t)^{\rT}] = 2i\min(s,t)J$ for all $s,t\>0$.  
The 
matrices $\sA$, $\sB $ in (\ref{dX}) are not arbitrary and satisfy the physical realizability (PR)  condition \cite{JNP_2008,SP_2012}
\begin{equation}
\label{PR}
  \sA \Theta + \Theta \sA^{\rT} + \sB J\sB ^{\rT} = 0,
\end{equation}
which pertains to the preservation of the CCRs (\ref{Theta}) in time. Such matrices are parameterized as
$  \sA = 2\Theta (K + M^\rT J M)$,
$
  \sB  = 2\Theta M^{\rT}
$
in terms of the energy and coupling matrices $K = K^{\rT} \in \mR^{n\x n}$, $M \in \mR^{m\x n}$, which specify the quadratic  system Hamiltonian $    \frac{1}{2}
    X^{\rT} KX
$ and the vector $MX$ of system-field coupling operators.

The relations (\ref{Theta})--(\ref{PR}), which describe the OQHO, reflect the effect of the external bosonic fields on its dynamics. Accordingly,  the system-field Hilbert space is organised as the tensor product
$    \fH := \fH_0 \ox \fF
$,
where $\fH_0$ is a Hilbert space for the action of the initial system variables $X_1(0), \ldots, X_n(0)$. The space $\fH$ is endowed with a filtration $(\fH_t)_{t\> 0}$, where  $\fH_t:= \fH_0 \ox \fF_t$, and $(\fF_t)_{t\> 0}$ is the Fock space filtration. At any time $t\> 0$,  the system variables $X_j(t)$ act on the subspace $\fH_t$ for all $j=1, \ldots, n$, while the input field variables $W_k(t)$ act on the subspace $\fF_t$ for all $k=1, \ldots, m$, in which sense both sets of processes (and nonanticipative functions thereof) are adapted to the filtration $(\fH_t)_{t\> 0}$. 
The statistical properties of the system and field variables depend on a density operator (quantum state) $\rho$ (a positive semi-definite self-adjoint operator on $\fH$ with unit trace $\Tr \rho = 1$) which also has a tensor-product structure:
$  \rho
  :=
  \rho_0 \ox \ups
$,
where $\rho_0$ is the initial system state on $\fH_0$, and the fields are in the vacuum state $\ups$ \cite{HP_1984,P_1992}. In particular, $\rho$  specifies the expectation $    \bE \xi
    :=
    \Tr(\rho \xi)
$ for quantum variables $\xi$ on the space $\fH$.

Since the solution of the linear QSDE (\ref{dX}) satisfies
$    X(t)
    =
    \re^{(t-s)\sA} X(s) + \int_s^t \re^{(t-\tau)\sA} \sB  \rd W(\tau)
$
for all     $
    t\> s\> 0$, and the future Ito increments of the quantum Wiener process $W$ commute with the past system variables (so that $[\rd W(\tau), X(s)^\rT] = 0$ for all $\tau\> s\> 0$), then $[X(t), X(s)^\rT] = \re^{(t-s)\sA} [X(s), X(s)^\rT]$. Hence,
the CCRs (\ref{Theta}), which are concerned with one point in time, extend to different moments as
\begin{align}
\label{XsXtcomm}
    [X(s), X(t)^{\rT}]
    & =
    2i\Lambda(s-t),
    \qquad
    s,t\> 0,\\
\label{Lambda}
    \Lambda(\tau)
    & =
    -\Lambda(-\tau)^\rT
    =
    \left\{
    \begin{matrix}
    \re^{\tau \sA}\Theta & {\rm if}\  \tau\> 0\\
    \Theta \re^{-\tau \sA^{\rT}} & {\rm if}\  \tau< 0\\
    \end{matrix}
    \right.,
\end{align}
where $\Lambda$ is the two-point CCR matrix of the system variables, with $\Lambda(0) = \Theta$. The linear structure of the QSDE (\ref{dX}) enters (\ref{Lambda}) through the matrix $\sA$, which is assumed to be Hurwitz.

\section{Quantum risk-sensitive cost functionals}
\label{sec:funs}
\setcounter{equation}{0}

The original quantum risk-sensitive cost functional 
\cite{J_2004,J_2005} employs an auxiliary quantum process in the form of the (leftward)  TOE
\begin{equation}
\label{Rt}
    R_{\theta}(t)
    :=
    \lexp
    \Big(
    \frac{\theta}{2}
    \int_0^t
    \Sigma(s)
    \rd s
    \Big),
    \qquad
    t\>0,
\end{equation}
which is the fundamental solution of the operator differential equation (ODE)
\begin{equation}
\label{Rdot}
    \dot{R}_{\theta}(t)
    =
    \tfrac{\theta }{2}
    \Sigma(t)
    R_{\theta}(t),
    \qquad
    R_{\theta}(0) = \cI_{\fH}.
\end{equation}
Here, $\dot{(\, )}:= \d_t(\cdot)$ is the time derivative, $\theta\>0$ is the risk-sensitivity  parameter, and
$\Sigma(t)$ is a time-dependent positive semi-definite
self-adjoint
quantum variable which can be a function (for example, quadratic)  of the current system variables (or, more generally, their past history over the time interval $[0,t]$), so that $\Sigma$ is an adapted quantum process. Since, in general,  $R_\theta(t)$ is a non-Hermitian  operator with a complex mean value, its mean square is used instead as a cost functional
\begin{equation}
\label{Et0}
    E_{\theta}(t)
    :=
    \bE
    (
        R_{\theta}(t)^{\dagger}
        R_{\theta}(t)
    )
\end{equation}
(with $(\cdot)^\dagger$ the operator adjoint),
which imposes an exponential penalty on the system variables through $\Sigma$ due to the multiplicative structure of the TOE $R_\theta$, with $\theta$ controlling its severity.  For simplicity, we do not include an additional terminal cost (on the time interval $[0,t]$) in (\ref{Et0}); cf. \cite[Eqs. (19)--(21)]{J_2005}.

If
the quantum variables $\Sigma(s)$ commuted  with each other for all $0\< s\< t$,  then (\ref{Et0}) would reduce to
\begin{equation}
\label{Eclass}
    E_{\theta}(t)
    =
    \bE
    \re^{
        \theta
            \int_0^t \Sigma(s)\rd s
    },
\end{equation}
which is organised as the classical exponential-of-integral performance criteria \cite{BV_1985,J_1973,W_1981}.
In the noncommutative quantum setting, the right-hand side of (\ref{Eclass}) provides an alternative to the original quantum risk-sensitive cost functional in (\ref{Rt}), (\ref{Et0}). Its quadratic-exponential counterpart \cite{VPJ_2018a} is given by
\begin{equation}
\label{QEF}
    \Xi_{\theta}(t)
    :=
    \bE \re^{\theta\varphi(t)}
    =
    \bE\re^{\theta \int_0^t \psi(s)\rd s},
\end{equation}
where $\varphi$ is a quantum process defined for any time $t\> 0$ by
\begin{align}
\label{phipsi}
    \varphi(t)
    :=
    \int_0^t
    \psi(s)
    \rd s,
    \qquad
    \psi(s)
    :=
    X(s)^{\rT} \Pi X(s).
\end{align}
Here, $\Pi$ is a real positive semi-definite symmetric matrix of order $n$ (the dependence of $\Xi_{\theta}(t)$ on $\Pi$ is omitted for brevity). Accordingly, $\varphi(t)$, $\psi(t)$ are positive semi-definite self-adjoint operators on the system-field space $\fH$, which follows from the representation
$    \psi
    =
    \zeta^{\rT}\zeta
    =
    \sum_{k=1}^n
    \zeta_k^2
$
in terms of the auxiliary self-adjoint quantum variables constituting the vector
$    \zeta
    :=
    (\zeta_k)_{1\< k \< n}
    :=
    \sqrt{\Pi} X
$.

Although the original quantum risk-sensitive cost $E_{\theta}$ in (\ref{Rt}), (\ref{Et0}) and its quadratic-exponential counterpart $\Xi_{\theta}$ in (\ref{QEF}), (\ref{phipsi}) 
are identical 
in the classical case if $\Sigma=\psi$, they are different in the noncommutative quantum setting (even if $\Sigma = \psi$) because of the discrepancy between the TOE and the usual operator exponential.
Moreover, at any given instant $t\>0$, the QEF $\Xi_{\theta}(t)$ is the moment-generating function 
for the classical probability distribution (the averaged spectral measure \cite{H_2001})  of the self-adjoint quantum variable $\varphi(t)$. In contrast to $\Xi_{\theta}(t)$, the quantity $E_\theta(t)$ in (\ref{Et0}) 
does not lend itself to a similar  association with 
a single $\theta$-independent quantum variable.

Since the evolution equations for the cost functionals (\ref{Et0}), (\ref{QEF}) are obtained by averaging the corresponding time derivatives as
$    \dot{E}_\theta
    =
    \bE
    (
        (R_{\theta}^{\dagger}
        R_{\theta})^{^\centerdot}
    )$ and
$
    \dot{\Xi}_\theta
    =
    \bE
    (
    (\re^{\theta\varphi})^{^\centerdot}
        )$, 
we will be concerned mainly with the dynamics of the processes  $R_{\theta}^{\dagger}
        R_{\theta}$ and $\re^{\theta\varphi}$ themselves. 
        Also, we will abandon the assumption on self-adjointness of the operator $\Sigma(t)$ which drives (\ref{Rdot}). Then an appropriate modification of \cite{VPJ_2018b} yields
\begin{equation}
\label{RRdot}
        (R^{\dagger}R)^{^\centerdot}
        =
        \tfrac{\theta}{2}
        (
        (\Sigma R)^\dagger R + R^\dagger \Sigma R
        )
        =
      \theta R^\dagger (\Re \Sigma) R,
\end{equation}
where the subscript $\theta$ in $R_\theta$ is omitted for brevity, and the real part is extended 
to operators as $\Re \xi:= \frac{1}{2}(\xi+\xi^{\dagger})$. Also,
\begin{equation}
\label{ephidot}
    (
        \re^{\theta \varphi}
    )^{^\centerdot}
    =
    \theta
    \re^{\frac{\theta}{2} \varphi}
    \Psi_\theta
    \re^{\frac{\theta}{2} \varphi},
    \quad
    \Psi_\theta
    :=
        \sinhc
    \big(
        \tfrac{\theta}{2}
        \ad_{\varphi}
    \big)
    (\psi),
\end{equation}
with $\ad_\xi(\cdot) := [\xi, \cdot]$,
where the evaluation of the hyperbolic  sinc function    $\sinhc(z)    := \sinc(-iz)$ at $\frac{\theta}{2}\ad_{\varphi}$ yields a linear superoperator acting on $\psi$. 
The relation (\ref{ephidot}) holds regardless of the particular structure of the OQHO dynamics and the processes in (\ref{phipsi}) (except that $\dot{\varphi}=\psi$) and follows from the identities
\begin{equation}
\label{Magnus}
    (\re^{\phi})^{^\centerdot}
    =
    \Ups(\ad_\phi)(\dot{\phi})
    \re^\phi
    =
    \re^\phi
    \Ups(-\ad_\phi)(\dot{\phi})
\end{equation}
(in view of the Magnus lemma \cite{M_1954})
for a time-varying operator $\phi$, which reduce to the standard exponential derivative when $[\phi, \dot{\phi}] = 0$,
where
\begin{equation}
\label{Ups}
    \Ups(z)
    :=
    \re^{\frac{z}{2}}\sinhc \tfrac{z}{2}
    =
    \left\{
    {\small\begin{matrix}1 & {\rm if}\ z= 0\\
    \tfrac{\re^z - 1}{z} & {\rm otherwise}\end{matrix}}
    \right..
\end{equation}
Therefore, 
the processes $R_\theta ^\dagger R_\theta$ and $\re^{\theta \varphi}$ 
reproduce each other (in which case, $R_\theta$ is a non-Hermitian operator square root of $\re^{\theta \varphi}$) if $\Re \Sigma$ in (\ref{RRdot}) is appropriately matched (and becomes unitarily equivalent) to $\Psi_\theta$ in (\ref{ephidot}), similarly to \cite[Theorem 3]{VPJ_2018b}. This suggests a link between the TOE-based quantum risk-sensitive functionals (\ref{Et0})  and the QEFs (\ref{QEF}), which   requires a more explicit representation of the process $\Psi_\theta$. To this end, the two-point CCRs (\ref{XsXtcomm}) and the specific quadratic dependence of $\varphi$, $\psi$ on the past history of the system variables lead to
\begin{align}
\nonumber
  \Psi_{\theta}(t)
  := &
\psi(t)
  +
    \tfrac{\theta}{2}
    \Big(
        \Re
        \Big(
            X(t)^{\rT}
            \int_0^t
            \alpha_{\theta,t}(\sigma)
            X(\sigma)
            \rd \sigma
        \Big)\\
\label{Psi}
        &+
        \int_{[0,t]^2}
        X(\sigma)^{\rT}
        \beta_{\theta,t}(\sigma,\tau)
        X(\tau)
        \rd \sigma
        \rd \tau
    \Big),
\end{align}
which is a quadratic function of the past history of the system variables over the time interval $[0,t]$, with the  functions $\alpha_{\theta, t}: [0,t]\to \mR^{n\x n}$, $\beta_{\theta,t}: [0,t]^2\to \mR^{n\x n}$ being related to the two-point CCR matrix $\Lambda$ in (\ref{Lambda}), with $\beta_{\theta,t}$ being symmetric:  $\beta_{\theta,t}(\sigma,\tau) = \beta_{\theta,t}(\tau,\sigma)^\rT$.  These kernel functions are obtained in \cite[Theorem~1, Lemma~2]{VPJ_2018a} using the fact that quadratic forms in quantum variables with CCRs
  form a Lie algebra with respect to the commutator
   (see, for example,  \cite[Appendix A]{VPJ_2018a} and references therein).

\section{A class of quadratic functions of system variables}
\label{sec:quad}
\setcounter{equation}{0}

In view of the structure of the right-hand side of (\ref{Psi}), consider the following unified representation for a class of quadratic functions of the  system variables of the OQHO.
Let $Q: \fB_+^2\to \mC^{n\x n}$ be a countably additive measure of bounded total variation on the $\sigma$-algebra $\fB_+^2$ of Borel subsets of the orthant $\mR_+^2$ (with $\mR_+:= [0,+\infty)$ the set of nonnegative real numbers). With any such $Q$, we associate  a quantum variable
\begin{equation}
\label{phiQ}
    \phi_Q
    :=
    \int_{\mR_+^2}
    X(\sigma)^\rT
    Q(\rd \sigma \x \rd \tau)
    X(\tau),
\end{equation}
which is a quadratic function of the system variables. 
For example, $\varphi(t)$, $\psi(t)$ in  (\ref{phipsi}) and $\Psi_\theta(t)$ in (\ref{Psi}) are particular cases of (\ref{phiQ}), as discussed below.
Since we will be concerned with commutators of the quantum variables (\ref{phiQ}), then,  due to the two-point CCRs (\ref{XsXtcomm}), the kernel measure $Q$ can be assumed to be symmetric in the sense that $Q(A\x B) = Q(B\x A)^\rT$ for  any $A,B \in \fB_+$  
(such measures form a complex linear space, which we denote by $\cC_n$).

Indeed, in view of the two-point CCRs (\ref{XsXtcomm}), for any \emph{antisymmetric} $\mC^{n \x n}$-valued measure $Q:=(q_{jk})_{1\< j,k\< n}$ on $\mR_+^2$ (with $Q(A\x B) = -Q(B\x A)^\rT$ for  any $A,B \in \fB_+$), the quantum variable (\ref{phiQ}) is a scalar:
$
    \phi_Q
     =
    \int_{\mR_+^2}
    \sum_{j,k=1}^n
    X_j(s)X_k(t)
    q_{jk}(\rd s \x \rd t)
    =
    \int_{\mR_+^2}
    \sum_{j,k=1}^n
    X_k(t)X_j(s)$ $
    q_{kj}(\rd t \x \rd s)$$
    =
    -
    \int_{\mR_+^2}
    \sum_{j,k=1}^n
    X_k(t)X_j(s)$    $
    q_{jk}(\rd s \x \rd t)=
    -
    \int_{\mR_+^2}
    \sum_{j,k=1}^n
    (X_j(s)X_k(t)- [X_j(s),X_k(t)])$
    $
    q_{jk}(\rd s \x \rd t)=
    -
    \phi_Q
    +
    2i
    \int_{\mR_+^2}
    \sum_{j,k=1}^n
    \lambda_{jk}(s-t)
    q_{jk}(\rd s \x \rd t)
    =
    i
    \int_{\mR_+^2}
    \bra
        \Lambda(s-t),
        Q(\rd s \x \rd t)
    \ket
$, 
which does not contribute to the commutators involving $\phi_Q$.  Here, $\bra M, N\ket:= \Tr(M^* N)$ is the Frobenius inner product of complex matrices, and    $\lambda_{jk}$ are the entries of the two-point CCR matrix $\Lambda$ in (\ref{Lambda}).  Therefore, since any $\mC^{n\x n}$-valued measure $Q$ on $\fB_+^2$ splits into symmetric and antisymmetric parts $Q_+$, $Q_-$ as
$
    Q_{\pm}(A\x B)
     :=
    \frac{1}{2}(Q(A\x B) \pm Q(B\x A)^\rT)
$,
then
$
    \phi_Q
    =
    \phi_{Q_+}+ \phi_{Q_-}$ $
    =
    \phi_{Q_+}
    +
    i
    \int_{\mR_+^2}
    \bra
        \Lambda(s-t),
        Q_-(\rd s \x \rd t)
    \ket
$
coincides with $\phi_{Q_+}$ up to an additive constant which is irrelevant for the commutators.

Now, for any $Q\in \cC_n$, its pointwise real and imaginary parts $\Re Q$, $\Im Q$ are symmetric $\mR^{n\x n}$-valued measures on $\fB_+^2$ (we denote the real linear space of such measures by $\cR_n$, so that $\cC_n = \cR_n + i\cR_n$), giving rise to the decomposition
\begin{equation}
\label{phiQReIm}
    \phi_Q = \phi_{\Re Q}  + i \phi_{\Im Q},
\end{equation}
where $\phi_{\Re Q}$, $\phi_{\Im Q}$ are self-adjoint quantum variables. Hence,
\begin{equation}
\label{phiQdagger}
    \phi_Q^\dagger =
    \phi_{\Re Q}  - i \phi_{\Im Q}
    =
    \phi_{\overline{Q}},
\end{equation}
with $\overline{Q}$ the pointwise complex conjugate of the measure $Q$. In accordance with (\ref{phiQReIm}), (\ref{phiQdagger}), any $Q \in \cR_n$ yields a self-adjoint quantum variable $\phi_Q$.

Also, we define the product of a measure $Q\in \cC_n$ and the two-point CCR function $\Lambda$ in (\ref{XsXtcomm}) as a function $\Lambda Q: \mR_+\x \fB_+ \to \mC^{n\x n}$ (which is a measure over its second argument) given by
\begin{equation}
\label{LambdaQ}
    (\Lambda Q) (t,B)
    :=
    \int_{\mR_+}
    \Lambda(t-\sigma) Q(\rd \sigma \x B)
\end{equation}
for all $t\>0$, $B\in \fB_+$. The function $\Lambda Q$ specifies the kernel of a linear integral operator which maps a function $f$ on $\mR_+$ with values in $\mC^n$ (or the space of vectors of $n$ quantum variables on $\fH$) to a function $g:= (\Lambda Q)(f)$ (of the same nature)  as
\begin{equation}
\label{fg}
    g(t)
    :=
    \int_{\mR_+^2}
    \Lambda(t-\sigma)
    Q(\rd \sigma \x \rd \tau)
    f(\tau),
    \qquad
    t\> 0.
\end{equation}
This integral operator corresponds to complex Hamiltonian matrices. In order to emphasize this analogy, $\Lambda Q$ will be referred to as a \emph{complex Hamiltonian  kernel} (CHK) (in the sense of the symplectic structure specified by $\Lambda$). CHKs are infinitesimal generators of \emph{complex symplectic kernels} (CSKs) $S: \mR_+ \x \fB_+ \to \mC^{n\x n}$ (which are also measures over the second argument) satisfying
\begin{equation}
\label{compsymp}
    \int_{\mR_+^2}
    S(s,\rd \sigma)
    \Lambda(\sigma-\tau)
    S(t,\rd \tau)^\rT
    =
    \Lambda(s-t),
    \ \
    s,t\>0.
\end{equation}
Such kernels $S$  form a semigroup, which preserves the two-point CCRs (\ref{XsXtcomm}) in the sense that the latter are inherited by the quantum process
$
    \wt{X}(t)
    :=
    \int_{\mR_+}
    S(t,\rd \sigma)
    X(\sigma)
$
as $    [\wt{X}(s),\wt{X}(t)^\rT]
     =
    \int_{\mR_+^2}
    S(s,\rd \sigma)
    [X(\sigma), X(\tau)^\rT]
    S(t,\rd \tau)^\rT
     =
    2i
    \int_{\mR_+^2}
    S(s,\rd \sigma)
    \Lambda(\sigma-\tau)
    S(t,\rd \tau)^\rT
    =
    2i
    \Lambda(s-t)$ for all $
    s,t\> 0
$. 

\section{Lie-algebraic isomorphism to complex Hamiltonian kernels}\label{sec:iso}
\setcounter{equation}{0}

The significance of the CHK $\Lambda Q$ in  (\ref{LambdaQ}), (\ref{fg}) for commutation relations is clarified by
\begin{align}
\nonumber
    [\phi_Q, X(t)]
     = &
    \int_{\mR_+^2}
    [X(\sigma)^\rT
    Q(\rd \sigma \x \rd \tau)
    X(\tau),\, X(t)
    ]    \\
\nonumber
    =&
    -
    \int_{\mR_+^2}
    [X(t),X(\sigma)^\rT]
    Q(\rd \sigma \x \rd \tau)
    X(\tau)\\
\nonumber
    & +
    \int_{\mR_+^2}
    \big(X(\sigma)^\rT
    Q(\rd \sigma \x \rd \tau)
    [X(\tau), X(t)^\rT]
    \big)^\rT    \\
\nonumber
    =&
    -2i
    \int_{\mR_+^2}
    \Lambda(t-\sigma)
    Q(\rd \sigma \x \rd \tau)
    X(\tau)\\
\nonumber
    & +
    2i
    \int_{\mR_+^2}
    \big(X(\sigma)^\rT
    Q(\rd \sigma \x \rd \tau)
    \Lambda(\tau - t)
    \big)^\rT    \\
\label{phiQXcomm}
    =&
    -4i
    \int_{\mR_+^2}
    \Lambda(t-\sigma)
    Q(\rd \sigma \x \rd \tau)
    X(\tau),
    \qquad
    t\> 0,
\end{align}
so that $[\phi_Q, X] = -4i(\Lambda Q)(X)$.
Here, the  derivation and antisymmetry properties of the commutator have been combined  with the antisymmetry of $\Lambda$ in (\ref{XsXtcomm}), (\ref{Lambda}) and the symmetry of $Q$.

\begin{lemma}
\label{lem:phiQcomm}
The quantum variables $\phi_Q$ in (\ref{phiQ}), associated with measures $Q\in \cC_n$, form a Lie algebra, in which
\begin{equation}
\label{phiQcomm}
    [\phi_{Q_1}, \phi_{Q_2}]
    =
    \phi_Q,
\end{equation}
where
\begin{equation}
\label{QQQ}
    Q = 4i (Q_1\Lambda Q_2 - Q_2 \Lambda Q_1)
\end{equation}
is also such a measure given by
\begin{align}
\nonumber
    Q(A\x B)
    = &
    4i
    \int_{\mR_+^2}
    \big(Q_1(A\x \rd s)\Lambda(s-t) Q_2(\rd t \x B)\\
\label{QQ}
    & - Q_2(A\x \rd s)\Lambda(s-t) Q_1(\rd t \x B)
    \big)
\end{align}
for all $A,B\in \fB_+$, where $\Lambda$ is the two-point CCR function from (\ref{XsXtcomm}).
\end{lemma}
\begin{proof}
 By a reasoning, similar to that in (\ref{phiQXcomm}), (\ref{phiQ}) implies
 \begin{align}
 \nonumber
    [\phi_{Q_1}, \phi_{Q_2}]
    = &
    \int_{\mR_+^2}
    [\phi_{Q_1},
    X(\sigma)^\rT
    Q_2(\rd \sigma \x \rd \tau)
    X(\tau)    ]\\
 \nonumber
    = &
    \int_{\mR_+^2}
    [\phi_{Q_1},
    X(\sigma)]^\rT
    Q_2(\rd \sigma \x \rd \tau)
    X(\tau)\\
 \nonumber
    & +
    \int_{\mR_+^2}
    X(\sigma)^\rT
    Q_2(\rd \sigma \x \rd \tau)
    [\phi_{Q_1},X(\tau)]\\
\nonumber
    = &
    -4i
    \int_{\mR_+^2}
    ((\Lambda Q_1)(X)(\sigma))^\rT
    Q_2(\rd \sigma \x \rd \tau)
    X(\tau)\\
 \nonumber
    & -
    4i
    \int_{\mR_+^2}
    X(\sigma)^\rT
    Q_2(\rd \sigma \x \rd \tau)
    (\Lambda Q_1)(X)(\tau)\\
 \label{phiQcomm1}
    = &
    \int_{\mR_+^2}
    X(\sigma)^\rT
    Q(\rd \sigma \x \rd \tau)
    X(\tau),
\end{align}
where $Q\in \cC_n$ is given by (\ref{QQ}), or, equivalently, (\ref{QQQ}), thus establishing (\ref{phiQcomm}). The symmetry of $Q$ follows from that of the measures $Q_1$, $Q_2$ and the antisymmetry of $\Lambda$. In (\ref{phiQcomm1}), use is also made of the relation
$    -\int_{\mR_+}
    ((\Lambda Q_1)(X)(\sigma))^\rT
        Q_2(\rd \sigma \x B)
         =-
    \int_{\mR_+^3}
    X(v)^\rT
    (\Lambda(\sigma-\tau)
    Q_1(\rd \tau \x \rd v))^\rT
        Q_2(\rd \sigma \x B)
         =
    \int_{\mR_+^3}
    X(v)^\rT
    Q_1(\rd v \x \rd \tau)
    \Lambda(\tau-\sigma)
        Q_2(\rd \sigma \x B)
         =
    \int_{\mR_+}
    X(v)^\rT
    (Q_1 \Lambda Q_2)(\rd v \x B)
$, 
where $Q_1 \Lambda Q_2$ is a $\mC^{n\x n}$-valued measure (not necessarily symmetric)  given by 
    $(Q_1 \Lambda Q_2)(A\x B)
    =
    \int_{\mR_+^2}
    Q_1(A \x \rd \sigma)
    \Lambda(\sigma-\tau)
        Q_2(\rd \tau \x B)$ 
for all $A,B\in \fB_+$.
\hfill$\blacksquare$\end{proof}

In accordance with (\ref{LambdaQ}), the multiplication of measures in $Q_1 \Lambda Q_2$ is associative.
From Lemma~\ref{lem:phiQcomm}, it follows that the Lie algebra of quantum variables $\phi_Q$ in (\ref{phiQ}), considered for measures $Q\in \cC_n$, is isomorphic to the Lie algebra of CHKs. Indeed, since (\ref{QQQ}) implies that
$
    4i\Lambda Q
    =
    (4i)^2 (\Lambda Q_1\Lambda Q_2 - \Lambda Q_2 \Lambda Q_1)
    =
    [4i\Lambda Q_1, 4i\Lambda Q_2]
$,
the Lie-algebraic isomorphism is described by the correspondence
\begin{equation}
\label{iso}
  \phi_Q \longleftrightarrow 4i\Lambda Q.
\end{equation}
Note that $Q\in \cC_n$ can be recovered from the two-sided Laplace transform of $\Lambda Q$ given by
\begin{align}
\nonumber
    \int_{\mR}
    \re^{-st}&
    \int_{\mR_+}
    \Lambda(t-\sigma)
    Q(\rd \sigma \x B)
    \rd t\\
\label{LambdaQLap}
    & =
    \wh{\Lambda}(s)
    \int_{\mR_+}
    \re^{-s\sigma}
    Q(\rd \sigma \x B)
\end{align}
in the strip $\{s\in \mC:\ 0 < \Re s < |\ln \br(\re^\sA)|\}$
for any $B \in \fB_+$, where $\br(\cdot)$ is the spectral radius of a square matrix, so that
  $\ln \br(\re^\sA) = \max_{1\< k\< n} \Re \lambda_k$, with $\lambda_1, \ldots, \lambda_n$ the eigenvalues of the Hurwitz matrix $\sA$. Here, the two-sided Laplace transform
\begin{align}
\nonumber
    \wh{\Lambda}(s)
    := &
    \int_{\mR}
    \re^{-st}
    \Lambda(t)
    \rd t
    =
    \int_{\mR_+}
    \big(
    \re^{-st}
    \re^{t \sA}
    \Theta
    +
    \re^{st}
    \Theta
    \re^{t \sA^{\rT}}
    \big)
    \rd t\\
\nonumber
    =&
    (sI_n-\sA)^{-1}\Theta
    -
    \Theta
    (sI_n + \sA^\rT)^{-1}\\
\nonumber
     =&
    (sI_n-\sA)^{-1}
    \big(
    \Theta(sI_n + \sA^\rT)
    -
    (sI_n-\sA)\Theta
    \big)
    (sI_n + \sA^\rT)^{-1}\\
\nonumber
    =&
    (sI_n-\sA)^{-1}
    \big(
    \sA\Theta
    +
    \Theta \sA^\rT
    \big)
    (sI_n + \sA^\rT)^{-1}\\
\label{Lambdahat}
    \hskip-3mm=&
    -
    (sI_n-\sA)^{-1}
    \sB J\sB ^\rT
    (sI_n + \sA^\rT)^{-1}
\end{align}
is a rational function, which is obtained by using the matrix exponential structure of $\Lambda$ in (\ref{Lambda}) and the PR property (\ref{PR}) of the matrices $\sA$, $\sB $. Since $\sA$ is assumed to be Hurwitz, the integrals in (\ref{Lambdahat}) are convergent over the strip $|\Re s| < |\ln \br(\re^\sA)|$. A sufficient  condition for unique recoverability of $Q$ from $\Lambda Q$ using (\ref{LambdaQLap}) is $\det (\sB J\sB ^\rT)\ne 0$, for which it is necessary that $n\< m$.

\section{A Lie-algebraic correspondence between TOE-based and quadratic-exponential  functions of system variables}
\label{sec:cont}
\setcounter{equation}{0}

Similarly to the case \cite{VPJ_2018c} of products of quadratic-exponential functions of a finite number of quantum variables with CCRs, a combination of Dynkin's lemma \cite{D_1947} 
with the Lie-algebraic isomorphism (\ref{iso}) 
leads to
\begin{equation}
\label{expQQQ}
    \re^{\phi_{Q_1}}
    \re^{\phi_{Q_2}}
    =
    \re^{\phi_Q},
\end{equation}
where $Q_1, Q_2, Q \in \cC_n$ 
are related by the complex symplectic factorization:
\begin{equation}
\label{eQeQeQ}
    \re^{4i\Lambda Q_1}
    \re^{4i\Lambda Q_2}
    =
    \re^{4i\Lambda Q}.
\end{equation}
All three exponentials in (\ref{eQeQeQ}) are integral operators with CSKs in the sense of (\ref{compsymp}).
A continuous-product version of this representation formula is
\begin{equation}
\label{FG1}
    \lexp
    \Big(
    \int_0^t
    \phi_{F_s}
    \rd s
    \Big)
    =
    \re^{\phi_{G_t}}.
\end{equation}
Here, $F_t, G_t \in \cC_n$ are time-dependent measures satisfying
\begin{equation}
\label{FG2}
    \lexp
    \Big(
    4i
    \int_0^t
    \Lambda
    F_s
    \rd s
    \Big)
    =
    \re^{4i\Lambda G_t}
\end{equation}
for all $t\>0$, which is equivalent to the ODE
\begin{equation}
\label{expGdot}
    \big(
        \re^{4i\Lambda G_t}
    \big)^{^\centerdot}
    =
    4i\Lambda F_t
    \re^{4i\Lambda G_t},
    \qquad
    G_0 = 0.
\end{equation}
A similar representation holds for the rightward TOEs
\begin{equation}
\label{FG3}
    \rexp
    \Big(
    \int_0^t
    \phi_{F_s}
    \rd s
    \Big)
    =
    \re^{\phi_{G_t}},
    \
    \rexp
    \Big(
    4i
    \int_0^t
    \Lambda
    F_s
    \rd s
    \Big)
    =
    \re^{4i\Lambda G_t},
\end{equation}
in which case, (\ref{expGdot}) is replaced with
\begin{equation}
\label{expGdotright}
    \big(
        \re^{4i\Lambda G_t}
    \big)^{^\centerdot}
    =
    4i
    \re^{4i\Lambda G_t}
    \Lambda F_t,
    \qquad
    G_0 = 0.
\end{equation}
The following theorem employs (\ref{expQQQ})--(\ref{expGdotright}) in order to relate two extended classes of functions of the OQHO variables whose averaging leads to the TOE-based and QEF costs in  (\ref{Et0}), (\ref{QEF}).


\begin{theorem}
\label{th:symp}
Suppose the quantum process $R$ in (\ref{Rt}), (\ref{Rdot})\footnote{the parameter $\theta$ is incorporated in the measures below, and the dependence on $\theta$ is omitted for brevity, or, equivalently, $\theta = 1$
} is driven as
\begin{equation}
\label{RF}
    R(t) := \lexp\Big(\frac{1}{2}\int_0^t \phi_{F_s}\rd s\Big)
\end{equation}
by the quantum variable (\ref{phiQ}) with a time-dependent measure $F_t\in \cC_n$. Then
\begin{equation}
\label{RRN}
     R(t)^\dagger R(t) = \re^{\phi_{N_t}},
\end{equation}
where $N_t \in \cR_n$ is a time-dependent measure, evolving as
\begin{equation}
\label{Ndot}
    \big(
    \re^{4i\Lambda N_t }
    \big)^{^\centerdot}
    =
    4i
    \re^{4i\Lambda \overline{G}_t }
    \Lambda (\Re F_t)
    \re^{4i\Lambda G_t },
    \qquad
    N_0=0,
\end{equation}
and $G_t \in \cC_n$ is a time-dependent measure governed by
\begin{equation}
\label{Gdot}
    \big(
    \re^{4i\Lambda G_t }
    \big)^{^\centerdot}
    =
    2i\Lambda F_t
    \re^{4i\Lambda G_t },
    \qquad
    G_0 = 0.
\end{equation}
\end{theorem}
\begin{proof}
By applying (\ref{FG1})--(\ref{expGdot}), it follows that the process $R$ in (\ref{RF}) can be represented as
\begin{equation}
\label{RG}
    R(t) = \re^{\phi_{G_t}},
\end{equation}
where the time-dependent measure $G_t\in \cC_n$ satisfies
\begin{equation}
\label{Glexp}
    \re^{4i\Lambda G_t} = \lexp \Big(2i \int_0^t \Lambda F_s \rd s \Big),
\end{equation}
which is equivalent to (\ref{Gdot}).  In view of (\ref{phiQdagger}), the adjoint of (\ref{RG}) takes the form
\begin{equation}
\label{RG+}
    R(t)^\dagger = \re^{\phi_{G_t}^\dagger} = \re^{\phi_{\overline{G}_t}}.
\end{equation}
By combining (\ref{RG}) with (\ref{RG+}) and using (\ref{expQQQ}), (\ref{eQeQeQ}), it follows that 
$    R(t)^\dagger R(t)
    =
    \re^{\phi_{\overline{G}_t}}
    \re^{\phi_{G_t}}
    =
    \re^{\phi_{N_t}}$, 
thus establishing (\ref{RRN}),
where $\phi_{N_t}$ is self-adjoint, and $N_t\in \cR_n$ satisfies the complex symplectic factorization
\begin{equation}
\label{NGG}
    \re^{4i\Lambda N_t}
    =
    \re^{4i\Lambda \overline{G}_t}
    \re^{4i\Lambda G_t}.
\end{equation}
On the other hand, (\ref{RF}), (\ref{phiQdagger}) 
imply that
\begin{equation}
\label{RF+}
    R(t)^\dagger
    =
    \rexp
    \Big(
        \frac{1}{2}
        \int_0^t
        \phi_{\overline{F}_s}
        \rd s
    \Big).
\end{equation}
Hence, application of (\ref{FG3}), (\ref{expGdotright}) to (\ref{RG+}), (\ref{RF+}) leads to
\begin{equation}
\label{Grexp}
    \re^{4i\Lambda \overline{G}_t}
    =
    \rexp \Big(2i \int_0^t \Lambda \overline{F}_s \rd s \Big).
\end{equation}
By substituting (\ref{Glexp}), (\ref{Grexp}) into (\ref{NGG}) and differentiating, it follows that
$    (\re^{4i\Lambda N_t})^{^\centerdot}
    =
    (\re^{4i\Lambda \overline{G}_t}
        )^{^\centerdot}
    \re^{4i\Lambda G_t}
    +
    \re^{4i\Lambda \overline{G}_t}
    (\re^{4i\Lambda G_t}
    )^{^\centerdot}
     =
    2i
    \re^{4i\Lambda \overline{G}_t}
    \Lambda (F_t + \overline{F}_t)
    \re^{4i\Lambda G_t}
    =
    4i
    \re^{4i\Lambda \overline{G}_t}
    \Lambda (\Re F_t )
    \re^{4i\Lambda G_t}
$, 
which proves (\ref{Ndot}), with 
$N_0 = 0$ due to $R(0) = \cI_\fH$.
\hfill$\blacksquare$\end{proof}

In view of the assumption $F_t \in \cC_n$ (rather than $F_t \in \cR_n$), the quantum variable $\phi_{F_t}$ in (\ref{RF}) is not necessarily self-adjoint, thus extending the original class of TOEs $R$ in \cite{J_2004,J_2005}. Another extension in Theorem~\ref{th:symp} is that the self-adjoint quantum processes $\phi_{N_t}$, specified by measures $N_t \in \cR_n$, contain $\varphi(t)$ in (\ref{phipsi}) as a particular case.

\section{Moving along the Lie-algebraic bridge}
\label{sec:move}
\setcounter{equation}{0}

Theorem~\ref{th:symp} allows
$N_t$ on the right-hand side of (\ref{RRN}) to be found for a given measure $F_t$ in (\ref{RF}), and the other way around,  $F_t$ can be found for a given $N_t$.

The first of these problems pertains to representing the TOE-based original quantum risk-sensitive cost functional as a QEF. An intermediate step of this procedure is concerned with finding the measure $G_t$ in (\ref{RG}) for the given $F_t$.
A comparison of the ODE (\ref{Gdot}) with the general exponential derivative
$
    (
    \re^{4i\Lambda G_t }
    )^{^\centerdot}
    =
    4i
    \Ups(4i \ad_{\Lambda G_t})
    (\Lambda \dot{G}_t)
    \re^{4i\Lambda G_t }
$
(following from (\ref{Magnus}))
leads to
\begin{equation}
\label{FG}
    \Ups
    (4i\ad_{\Lambda G_t})
    (\Lambda \dot{G}_t)
    =
    \tfrac{1}{2}
    \Lambda F_t,
\end{equation}
and hence,
\begin{equation}
\label{LambdaGdot}
    \Lambda \dot{G}_t
    =
    \tfrac{1}{2}
    \mho(4i\ad_{\Lambda G_t})(\Lambda F_t),
\end{equation}
where the function $\Ups$ is given by (\ref{Ups}), and its reciprocal 
$    \mho(z):= \frac{1}{\Ups(z)}
    =
    \sum_{k=0}^{+\infty}
    \frac{b_k}{k!}
    z^k$ 
is the generating function of the Bernoulli numbers \cite{A_1991}
$b_0, b_1, b_2, \ldots$. The relation (\ref{LambdaGdot}) is a nonlinear ODE whose linearised version takes the form 
$    \Lambda \dot{G}_t
    =
    \tfrac{1}{2}\Lambda F_t
    +
    i[\Lambda F_t,\Lambda G_t] + (*)$ 
in view of 
$b_0 = 1$, $b_1 = -\frac{1}{2}$, where $(*)$ contains the higher-order terms, nonlinear with respect to $G_t$.
However, finding $N_t$ from (\ref{NGG}) requires the CSK $S_t: \mR_+\x \fB_+ \to \mC^{n\x n}$ of the integral operator $\re^{4i\Lambda G_t}$ in (\ref{Glexp}) rather than $G_t$ itself. In contrast to (\ref{LambdaGdot}), $S_t$ satisfies a linear integro-differential equation (IDE)
\begin{equation}
\label{SIDE}
    \d_t S_t(v,B)
    =
    2i
    \int_{\mR_+^2}
    \Lambda(v-\sigma)
    F_t(\rd \sigma \x \rd \tau)
    S_t(\tau,B)
\end{equation}
for all $t,v\>0$, $B \in \fB_+$, with the initial condition $S_0(v,B) = \chi_B(v)I_n$, where $\chi_B$ is the indicator function of the set $B$. Then the measure $N_t$ is recovered from the CSK $T_t: \mR_+\x \fB_+ \to \mC^{n\x n}$ of the integral operator $\re^{4i \Lambda N_t}$ satisfying the complex symplectic factorization
\begin{equation}
\label{STS}
    \int_{\mR_+}
    \overline{S}_t(v,\rd \sigma)
    T_t(\sigma, B)
    =
    S_t(v,B).
\end{equation}
The latter is a  linear  equation (of Fredholm first kind) obtained from (\ref{NGG})  due to the property that $\overline{S}_t$ is the CSK of the integral operator $\re^{-4i\Lambda \overline{G}_t} = (\re^{4i\Lambda \overline{G}_t})^{-1}$.
 By a similar reasoning, the following IDE form of (\ref{STS}) for finding $T_t$ (after the IDE (\ref{SIDE}) is solved for $S_t$) is obtained from (\ref{Ndot}):
 \begin{align}
\nonumber
    \int_{\mR_+}&
    \overline{S}_t(v,\rd \sigma)
    \d_t T_t(\sigma,B)\\
\label{STSdot}
    & =
    4i
    \int_{\mR_+^2}
    \Lambda(v-\sigma)
    \Re F_t(\rd \sigma \x \rd \tau)
    S_t(\tau,B)
\end{align}
(with the same initial condition $T_0 = S_0$). Therefore, the representation of the TOE-based left-hand side of (\ref{RRN}) as a quadratic-exponential function of the OQHO variables on the right-hand side can be  carried out by consecutive solution of the IDEs (\ref{SIDE}), (\ref{STSdot}).

The inverse problem (to the above) is to represent the QEF, specified by a given measure $N_t\in \cR_n$,    in the form of the original quantum risk-sensitive functional driven by $\phi_{F_t}$, where $F_t$ is to be found for $N_t$. To this end, the measure $F_t \in \cC_n$ in Theorem~\ref{th:symp} can be organised so  that the TOE $R(t)$ remains  a positive definite self-adjoint square root of $\re^{\phi_{N_t}}$ over the course of time: $R(t) = \re^{\frac{1}{2}\varphi_{N_t}}$ for all  $t\>0$. Then the corresponding measure $G_t$ in (\ref{RG}) is given by $G_t = \frac{1}{2}N_t$, and its substitution into (\ref{FG}) relates $F_t$ to $N_t$ as
\begin{equation}
\label{FN}
    \Lambda F_t
    =
    \Ups(2i\ad_{\Lambda N_t})(\Lambda \dot{N}_t)
    =
    \int_0^1
    L_{\lambda, t}
    \rd \lambda.
\end{equation}
Here, $N_t$ is assumed to have an appropriate distributional time derivative \cite{V_2002}, and
\begin{equation}
\label{L}
  L_{\lambda, t}
  :=
  \re^{2\lambda i \ad_{\Lambda N_t}}
  (\Lambda \dot{N}_t)
\end{equation}
is a CHK satisfying
\begin{equation}
\label{Lder}
    \d_\lambda
    L_{\lambda, t}
    =
    2i
    [\Lambda N_t, L_{\lambda, t}],
    \qquad
    0\< \lambda \< 1,
\end{equation}
with the initial condition (in the sense of the parameter $\lambda$)  $L_{0, t} = \Lambda \dot{N}_t$. Therefore, the quadratic-exponential function of the OQHO variables on the right-hand side of (\ref{RRN}) can be  represented in the TOE-based form on the left-hand side of (\ref{RRN}) by solving the IDE (\ref{Lder}) and performing the integration in (\ref{FN}).

\section{Specific nonanticipative time-varying measures}
\label{sec:spec}
\setcounter{equation}{0}

The above problems in Section~\ref{sec:move} (of finding $N_t$ for $F_t$, and $F_t$ for $N_t$) 
are particularly important for \emph{nonanticipative} time-varying measures $Q_t \in \cC_n$ satisfying
\begin{equation}
\label{suppQ}
    \supp Q_t \subset [0,t]^2,
    \qquad
    t\>0.
\end{equation}
Then
$\phi_{Q_t} =     \int_{[0,t]^2}
X(\sigma)^\rT Q_t(\rd \sigma \x \rd \tau) X(\tau)$ in (\ref{phiQ}) depends only on the past history of the system  variables over the time interval $[0,t]$ (and is, therefore, $\fH_t$-adapted). Furthermore, in view of (\ref{suppQ}) and in accordance with (\ref{LambdaQ}),  the corresponding CHK $\Lambda Q_t$  takes the form 
$    (\Lambda Q_t)(v,B)
    =
    \int_{\mR_+}
    \Lambda(v-\sigma)
    Q_t(\rd \sigma \x B)
     =
    \int_0^t
    \Lambda(v-\sigma)
    Q_t(\rd \sigma \x ([0,t]\bigcap B))
$ 
for any $t,v\>0$, $B \in \fB_+$, and hence, its support (over the second argument) satisfies 
$    \supp (\Lambda Q_t)(v,\cdot)
    \subset
    [0,t]$ 
for any $t\>0$.

Nonanticipative measures specify the quantum processes $\varphi$, $\psi$ in (\ref{phipsi}) and also play a role when the process $\Sigma$, which drives the TOE   in (\ref{Rdot}), is a quadratic function of the current system variables.
More precisely, the operator $\varphi(t)$ in  (\ref{phipsi}), which gives rise to the QEF in (\ref{QEF}), is a particular case of (\ref{phiQ}) obtained as $\varphi(t) = \phi_{N_t}$ by using a nonanticipative measure $N_t$ given by
\begin{equation}
\label{diagmeas}
    N_t(C):=
    \mu
    \{
        \sigma\in[0,t]:\
        (\sigma, \sigma) \in C
    \}
    \Pi,
    \quad
    C \in \fB_+^2,
\end{equation}
where $\mu$ is the one-dimensional Lebesgue measure. The distributional time derivative
of (\ref{diagmeas}) is an atomic nonanticipative measure concentrated at the singleton $\{(t,t)\}$ as
\begin{equation}
\label{diagmeasdot}
    \dot{N}_t(C)
    =
    \chi_C((t,t))\Pi,
\end{equation}
which allows the quantum variable $\psi(t)$ in (\ref{phipsi}) to be represented in the form (\ref{phiQ}) as $\psi(t) = \phi_{\dot{N}_t}$. Substitution of (\ref{diagmeas}), (\ref{diagmeasdot}) into (\ref{Lder}) (which pertains to the problem of finding the TOE-based representation for the QEF) leads to the following IDE for the CHK $L_{\lambda, t}$ in (\ref{L}):
\begin{align}
\nonumber
    \d_\lambda
    L_{\lambda, t}(v,B)
    = &
    2i
    \Big(
    \int_0^t
    \Lambda(v-\sigma)
    \Pi L_{\lambda, t}(\sigma, B)
    \rd \sigma\\
\label{LIDE}
    & -
    \int_0^t
    L_{\lambda, t}(v,\rd\sigma)
    \int_{[0,t]\bigcap B}
    \Lambda(\sigma-\tau)
    \Pi
    \rd\tau
    \Big),
\end{align}
with the initial condition $L_{0, t}(v,B) = 
\chi_B(t) \Lambda(v-t)\Pi$ for all $v\> 0$, $B \in \fB_+$. Here, use is also made of the fact that the measure $N_t$ in (\ref{diagmeas})  satisfies $N_t(A\x B) = \mu([0,t]\bigcap A \bigcap B)\Pi$ for all $A, B \in \fB_+$.

Furthermore, the quantum process $\Psi_1$  in (\ref{ephidot}) (we let $\theta=1$ as mentioned above) can also be represented in the form (\ref{phiQ}) as
\begin{equation}
\label{PsiL}
  \Psi_1(t)
  =
    \sinhc
    (
        \tfrac{1}{2}
        \ad_{\phi_{N_t}}
    )
    (\phi_{\dot{N}_t})
    =
    \phi_{M_t}
\end{equation}
with a nonanticipative measure $M_t \in \cR_n$.  It follows from (\ref{Psi}) that
$M_t$ consists of an absolutely continuous part over the square $[0,t]^2$, a singular   part concentrated at the two edges $([0,t]\x \{t\}) \bigcup (\{t\}\x [0,t])$ of the square, including an atomic part concentrated at the corner $\{(t,t)\}$. Due to the Lie-algebraic isomorphism of Section~\ref{sec:iso}, the measure $M_t$ in (\ref{PsiL}) satisfies
$
    \sinhc
    (
        2i
        \ad_{\Lambda N_t}
    )
    (\Lambda \dot{N}_t)
    =
    \Lambda M_t
$.

We will now return to the first of the problems in Section~\ref{sec:move} on the Lie-algebraic correspondence of Theorem~\ref{th:symp} in application to representing  the TOE-based criterion as a QEF. Suppose the quantum process $R$ in (\ref{Rt}),  (\ref{Rdot}) (with $\theta=1$ for simplicity) is in the form (\ref{RF}), where the time-dependent measure $F_t \in \cR_n$ is given 
by
\begin{equation}
\label{FPi}
    F_t(C)
    :=
    \chi_C((t,t))\Pi,
    \qquad
    t\>0,
    \
    C \in \fB_+^2,
\end{equation}
which is identical to the right-hand side of (\ref{diagmeasdot}). 
This corresponds to $\Sigma(t) = \psi(t) = X(t)^\rT \Pi X(t)$ in view of (\ref{phipsi}).
Then the IDE (\ref{SIDE}) for the CSK $S_t$ of the integral operator $\re^{4i\Lambda G_t}$ in (\ref{Glexp}) is driven by the atomic measure (\ref{FPi}) and reduces to a PDE: 
\begin{equation}
\label{SPDE}
    \d_t S_t(v,B)
    =
    2
    i
    \Lambda(v-t)
    \Pi
    S_t(t,B),
    \qquad
    t, v\>0,\
    B \in \fB_+,
\end{equation}
with the same initial condition $S_0(v,B) = \chi_B(v) I_n$. The transformation
$
    \wt{S}_t(u,B) := S_t(t+u,B)
$
allows the PDE (\ref{SPDE}) to be represented as
\begin{equation}
\label{SPDE1}
    \d_t \wt{S}_t(u,B)
    =
    2
    i
    \Lambda(u)
    \Pi
    \wt{S}_t(0,B)
    +
    \d_u \wt{S}_t(u,B).
\end{equation}
The PDE (\ref{SPDE}) (or its equivalent form (\ref{SPDE1})) can be solved by the method of characteristics or the Laplace transform techniques. The latter  employ the two-sided Laplace transform (\ref{Lambdahat}) of the two-point CCR function (\ref{Lambda})   for the system variables and are also applicable to the IDE (\ref{LIDE}).

\section{Conclusion}
\label{sec:conc}
\setcounter{equation}{0}

For linear quantum stochastic systems, we have established a Lie-algebraic link between two classes of quantum risk-sensitive cost functionals, which pertain to the original TOE-based performance criterion and its recent QEF version.  We have used a unified representation for the quadratic functions of system variables  in these criteria in terms of complex symmetric matrix-valued measures. The Lie-algebraic correspondence has been reduced to IDEs for related complex Hamiltonian and symplectic kernels which involve the two-point CCR matrix of the system variables. These relations will be employed in subsequent publications for extending useful features, such as robustness properties, simplicity of evolution, and applicability of information state techniques, from one of the classes of risk-sensitive costs to the other. The results of the paper will also be used in order to develop  state-space equations for computation and minimization of these functionals in quantum robust control and filtering problems.

\end{document}